\documentclass{IEEEtran4PSCC}
\IEEEoverridecommandlockouts
\usepackage{cite}
\usepackage{amsmath,amssymb,amsfonts}
\usepackage{amsthm}
\usepackage{algorithmic}
\usepackage{graphicx}
\usepackage{textcomp}
\usepackage{color}
\usepackage{tikz}
\usepackage{mathtools}
\usepackage{pgfplots}
\usepackage{siunitx}
\usepackage{lipsum}
\usetikzlibrary{spy}
\usetikzlibrary{shapes,arrows}
\newtheorem{remark}{Remark}

\newtheorem{proposition}{Proposition}
\def\BibTeX{{\rm B\kern-.05em{\sc i\kern-.025em b}\kern-.08em
    T\kern-.1667em\lower.7ex\hbox{E}\kern-.125emX}}
\begin{document}
\title{On the damping of power systems during inverter-based blackstart 
	\thanks{A.Anta and D.Cifelli are with AIT Austrian Institute of Technology GmbH, Vienna, Austria, email: \{Adolfo.Anta, Diego.Cifelli\}@ait.ac.at.}% <-this % 
}
\author{Adolfo~Anta, Diego~Cifelli}

\maketitle

\begin{abstract}
The massive deployment of inverter-based generation poses several challenges to system operators but also offers new opportunities. In the context of grid recovery, inverter-based generation is expected to take over new responsibilities, as conventional generation that are blackstart capable are being decommissioned.  Inverter-based facilities, being very flexible and controllable, have the ability to contribute during grid restoration, although it is still unclear how to deal with the limited ratings of batteries, PV, etc. and how to address nonlinear effects such as inrush currents or harmonics. In this direction, the requirements of grid-forming inverters to comply with black start capabilities are still to be defined. During the first steps of a grid restoration process the dynamics of the grid largely vary from the grid behaviour under normal operation. In particular, the grid is nearly unloaded and therefore poorly damped. Indeed, overvoltages and resonances are a common concern during the energization of long transmission lines. In this work we analyse the spectral properties of unloaded transmission lines, and describe the impact of voltage controllers on the damping of the system. As a byproduct, this analysis can be leveraged to propose high level requirements for grid-forming devices, in order to provide adequate damping during the early stages of a grid restoration process. 
\end{abstract}

\begin{IEEEkeywords}
Grid restoration; grid-forming control; resonances; eigenvalue analysis; Toeplitz matrices.
\end{IEEEkeywords}

\section{Introduction and motivation}
The introduction of inverter-based generation at large scale is causing issues in the planning and operation of power grids. At the same time, it raises the opportunity to upgrade and optimize standard approaches using inverter-based resources, as inverters are much more flexible and controllable than conventional generation. This transition is also affecting ongoing strategies for grid restoration. Inverter-based generation have the ability to contribute during the black start~\cite{noris2019power,ng_restart,jain2020blackstart}, although it is still unclear how these inverters should behave, and how system-level strategies may need to be adapted.

There has been considerable recent work at the unit level to assess the black start capabilities of inverter-based generation. However, when it comes to system-level analysis, grid restoration has been widely studied from a static perspective, with a focus on designing topologies for each step of the blackstart. Even in the few cases where dynamic aspects are considered, the focus lies on grids dominated by conventional power plants, thus modelled using swing equations and ignoring line dynamics. In this work we focus instead on the dynamic aspects during the early stages of a grid restoration under inverter-based generation, and in particular we analyse the damping (or lack thereof) during the energization of transmission lines by an inverter-based unit. 

In terms of strategies, grid restoration typically follows a divide-and-conquer approach: a grid is divided in several regions that are independently restored by black-start capable power plants. This approach creates many regional constraints and faces the issue of resynchronization among the regions. A faster and simpler strategy consists in defining a backbone structure, as extensive as possible, that is energized at once. This approach is being tested already by system operators with promising results~\cite{aniceto2023towards}.  The backbone energization is performed by means of a single unit, typically using a so-called soft start (a ramp in voltage magnitude), that minimizes certain adverse effects such as inrush currents~\cite{jain2021blackstart}. During a backbone energization, the lack of loads connected in the grid implies a very low level of damping, and resonances and overvoltages are common and can jeopardize an adequate black start. This is especially the case under the presence of transformers, whose energization can lead to a significant harmonic content. In traditional black start approaches, loads may be connected in order to avoid these issues, given that conventional power plants are characterized by large ratings. Due to the limited ratings of inverters, it is likely that all loads need to be disconnected in order to allow the inverters to energize the grid. Moreover, the expected resonances are different from the ones at normal operation~\cite{saldana2001methodology} (and partially unknown due to the uncertainties during black start), hence designing damping controls targeting a particular frequency may be ineffective. In particular, the first resonant frequency can be lower than during normal operation, and therefore more detrimental~\cite{emin2014transformer}.

This paper addresses these issues by analysing the dynamics during the early stages of a black start, considering transmission lines using $\pi$-sections and simple models of grid-forming units. Unlike in previous works, we do not rely on numerical studies but rather develop an analytical framework to study the impact of the grid and different control parameters on the system damping. We first derive the dynamics of a transmission line under no load, and its similarities with tridiagonal 2-Toeplitz matrices are pointed out. These matrices have very special characteristics, and expressions for the eigenvalues can be analytically computed. Using these expressions, we later study the influence of certain control structures on the energizing units, where a clear trade-off between reference tracking and damping can be observed. Moreover, we show how adding virtual resistance loops can contribute to the overall damping of the system. We hope this analysis contributes towards the definition of high level specifications for black start functionality from grid-forming units. Indeed, customized black start control modes may be required during the energization. Most of the focus on grid-forming research has centered around frequency and angle dynamics, while for this problem the voltage magnitude is highly pertinent. 

The rest of the paper is structured as follows: first, we model a transmission line using $\pi$-sections, to be energized by controllable voltage sources, and construct a state space representation of the line. Next, we establish the connection between our dynamic model and the well-studied field of Toeplitz matrices and Chebyshev polynomials. Using this parallelism, we derive a set of nonlinear equations implicitly defining the eigenvalues of the system for different cases: in open loop, under a proportional voltage controller, and using a virtual resistance. Finally, an example shows the impact of the controller parameters on the damping of the system. We conclude the paper enumerating a list of possible future steps to better understand the grid dynamics during this early stage of the grid restoration, and thereby defining possible control structures for black-start capable grid-forming converters.

%Focus has been on the freq/angle aspects for GFMs, not so much on the voltage, except that here seems to be the main issue. 

%Most of the existing studies are based in numerical simulations.

%Partially we are interested in analysis more than in design, since each GFM might come from a different manufacturer

%Talk about overvoltages during restoration, especially in the presence of transformers~\cite{emin2014transformer,poujade2018simulation}

%It has been reported that GFMs under no load might see voltage oscillations, so the gains have to be reduced (https://www.youtube.com/watch?v=la7M6bcOOeM, HECO part)

%We actually could design dedicated controllers (or current limitation techniques) just for black start- would standard techniques actually work? HEre you see overcurrent in the unit that has a lower voltage, so lowering the voltage does not help in reducing the current

%What if we find some controllers that minimize the overshoot? But hold on, can I actually call it overshoot? It looked like that for Q (and with sync clocks for the angle dynamics), but not for current, that's no overshoot but rather large transient

%\cite{alassi2020performance}

%Current limitation schemes should not "interrupt" the ramp to guarantee a soft energization? Or anti-islanding triggering, seems to be calling for trouble to have each inverter "do as much as they can, till someone comes to help out" Most current limitation techniques would reduce the voltage, which might jeopardize the synchronization of grid-following inverters

\begin{figure*}[h]
%\hrule
\includegraphics[width=\linewidth]{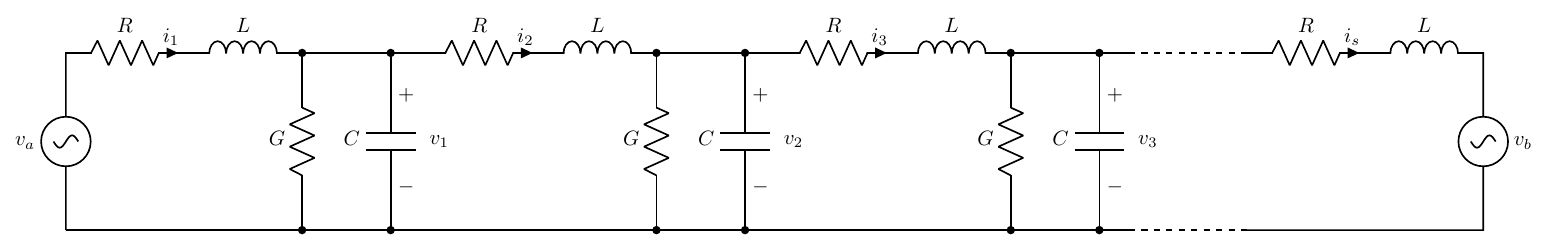}
\caption{Transmission line model with cascaded $\pi$ networks}
\label{fig:PiModel}
\end{figure*}

\section{Modelling}
We consider a transmission line (see Figure~\ref{fig:PiModel}) representing the grid backbone, to be modelled by means of a cascaded connection of $\pi$-sections. Unlike distributed parameter line models, these sections allow us to represent the dynamics through a linear state space model and a finite number of states~\cite{macias2005comparison}. For simplicity, a single phase system is selected. At both ends of the line voltage sources, representing grid-forming devices, are responsible for the energization of the line. Writing one differential equation per dynamic element (capacitors and inductances) leads us to a standard linear representation of the type $\dot x = A x +B u$ such as in~\eqref{eq:modelPi}, with $R$, $G$, $L$ and $C$ representing the resistance, conductance inductance and capacitance per $\pi$-section, $i$ and $v$ denoting the current and the voltage for each branch and node, and $v_a$, $v_b$ the voltage sources. Non-mentioned entries in the $A$ matrix in~\eqref{eq:modelPi} are equal to zero. %The subscript $1$ corresponds to the beginning of the line and $n$ denotes the end of the line. %, connected to the voltage source responsible for the energization of the line. 

The $A$ matrix in the state space representation is not only a tridiagonal matrix, but in particular it is consistent with the structure of a 2-Toeplitz matrix, for which eigenvalues can be analytically computed, as explained in the next section. Grid dynamics have been widely studied, but in general with the focus on analyzing stability properties~\cite{dorfler2018electrical}. At the same time, damping in power grids has also been previously explored~\cite{mallada2011improving,borsche2015effects}, but ignoring line dynamics and focusing on the dynamics of conventional generation. It has been conjectured as well that resonant frequencies and amplitudes depend on the short circuit power, the degree of meshing and amount of connected loads. However, a theoretical understanding of these design choices is missing. Moreover, it is unclear whether this intuition still holds when we replace synchronous machines with power converters.

%\rr{HOLD ON A SEC - I CANNOT JUST HAVE ONE LINE OPEN, THERE WOULD BE NO CURRENT AND THUS NO ODE. I EITHER WRITE THE SYSTEM WITH A) TWO VOLTAGE SOURCES, OR B) I LOOK HOW THE EXPRESSIONS LOOK FOR N=EVEN, OR C) HOW IT LOOKS WHEN I ADD AN R AT THE END (THAT SHOULD BE DOABLE?)}
%PROB BEST TO FIRST JUST ADD TWO VOLTAGE SOURCES

%\onecolumn
\begin{figure*}[h]
\begin{align}
\frac{d}{dt} \left[ \begin{array}{c}
i_1  \\ v_1 \\ i_2 \\ \vdots \\ v_{s-1} \\ i_s   
\end{array}\right]&=
\left[\begin{array}{rrrrrr}
-R/L & -1/L &        &        &        &\\
1/C  & -G/C & -1/C   &        &        &\\
     &  1/L & -R/L   & -1/L   &        &\\
     &      & \ddots & \ddots & \ddots &\\
     &      &        & 1/C    &  -G/C  & -1/C\\
     &      &        &        &   1/L  &  -R/L
\end{array}\right] 
\left[ \begin{array}{c}
i_1  \\ v_1 \\ i_2 \\ \vdots \\ v_{s-1} \\ i_s   
\end{array}\right] +
\left[\begin{array}{cccccc}
1/L &  0\\
0  & 0\\
0  & 0\\
\vdots & \vdots\\
0  & 0\\
0 &   -1/L
\end{array}\right] 
\left[ \begin{array}{c} v_a  \\ v_b   
\end{array}\right]
\label{eq:modelPi}
\end{align}
\end{figure*}
The controllable voltage sources $v_a, v_b$ track a certain voltage reference, either a step or a ramp. We start with a simple proportional controller, for instance:
\begin{equation}
    v_a(t) = k_V (v_{ref_A}-v_1), \:\:    v_b(t) = k_V (v_{ref_B}-v_{s-1}) \label{eq:propCtrlVoltage}
\end{equation}
where $k_V>0$ is the proportional gain and $v_1, v_{s-1}$ are the voltages at their local node, which is the typical measurement of the inverter at the point of common coupling. The reference voltage $v_{ref}$ for the energizing unit(s) during black start is typically selected to be a ramp, in order to minimize the impact of resonances, overvoltages and inrush phenomena~\cite{jain2021blackstart}.  A proportional-integral (PI) voltage controller may not be adequate since it may demand unnecessary reactive power, and perfect tracking is not necessary. Likewise, there is typically no clear reactive power reference to be used for the voltage regulation scheme in~\eqref{eq:propCtrlVoltage}, since the the reactive power demand is unclear during black start. High control gains would be prefered to guarantee adequate tracking, but this may reduce overall damping, as we will see in the next section. Moreover, should this soft energization be carried out by several units simultaneously, it is necessary to avoid having a very aggressive regulation at each unit that may lead to large power flows between the energizing units. 

While in some cases the ramp may take in the order of seconds~\cite{aniceto2023towards}, it is also common to slow down the ramps to take up to several minutes due to concerns of possible resonances. Slowing down has two drawbacks: first, protections may be triggered as the voltage rises very slowly. Second, it slows down the whole grid restoration process. It is also unclear how slow the ramp needs to be in order to avoid resonances. Ideally, we would find improved methods to tackle possible resonances, especially given the flexibility and controllability of power converters.

%\twocolumn
\begin{remark}
In this paper we model the energizing unit as an ideal, controllable voltage source, since the focus of this work is not on the internal controls of grid-forming units but rather to describe the expected high level behaviour of the grid-forming unit during grid restoration, with the goal of defining a possible high level specification for black start capable inverters. This is typically a valid assumption if the AC and DC sides of the converter are fairly decoupled~\cite{tayyebi2022system}.
%Nonetheless, grid-forming inverters are typically equipped with current loops and other control mechanisms that may interact with line dynamics~\cite{markovic2021understanding}. 
\end{remark}
% In which paper they had the same C and G for each section? I forgot... I think it was Salvatore D'Arco and Jef from KU Leuven? Nope

\section{Spectral analysis and analytical characterization}

In order to analyze the eigenvalues of the system, we draw in this section parallelisms between the state space representation in~\eqref{eq:modelPi} and Toeplitz matrices. This helps us deriving analytical expressions for the spectrum and its dependence on the voltage source and its controller gains. 

\subsection{An overview to 2-Toeplitz matrices and Chebyshev polynomials}

A matrix $A$ is tridiagonal if $a_{ij}=0$ whenever $\vert i-j\vert>1$. Moreover, $A$ is r-Toeplitz of order $n$ if $a_{i+r,j+r}=a_{ij}$, for $i,j=1,2,...n-r$~\cite{gover1994eigenproblem}. With this definition, it can be seen that the state space matrix $A$ in~\eqref{eq:modelPi} is a 2-Toeplitz matrix. The inverse and spectra of these matrices have been widely studied in mathematics, and has many applications in different topics such as wave dispersion models or fiber optic design~\cite{al2017pole,bastawrous2022closed}. In this section we briefly summarize existing results on $2$-Toeplitz matrices that will be helpful in computing analytical expressions for the eigenvalues of the $A$ matrix in~\eqref{eq:modelPi}. 

The characteristic polynomial of these matrices can be written as a function of Chebyshev polynomials. We first define the Chebyshev polynomials of second kind, that satisfy the following 3-point recurrent relationship: 
\begin{equation}
    U_{n+1}(\lambda) := \lambda U_{n}(\lambda) -U_{n-1}(\lambda),\qquad \forall n=1,2,\hdots \:\: \lambda \in\mathbb{C}
\end{equation}
with initial conditions $U_0(\lambda)=1$ and $U_1(\lambda)=2\lambda$. Moreover, $U_n$ satisfies:
\begin{equation}
    U_{n}(\lambda) = \frac{\sin ((n+1)\theta)}{\sin (\theta)}, \text{with } \lambda=\cos(\theta), \: 0\leq\theta < \pi \label{eq:trigDefChebyshev}
\end{equation}
For convenience we now define the following functions:
\begin{eqnarray}    
%g(\lambda) &:=& (\lambda-a_1)(\lambda-a_2), \\
%P_n(x) &:=&(b_1 b_2)^{(n/2)} U_n\left(\frac{x-b1-b2}{2\sqrt{b_1 b_2}}\right)\\
g(\lambda) &:=& \left(\lambda+\frac{R}{L}\right) \left(\lambda+\frac{G}{C}\right)     \label{eq:defG}\\
P_n(\lambda) &:=&\frac{1}{\left(LC\right)^{n}} U_n\left(\frac{1}{2}LC\lambda+1\right) \label{eq:defP}
\end{eqnarray}
With these definitions, we can already state that the eigenvalues of the $A$ matrix in~\eqref{eq:modelPi} correspond to the solution of the characteristic polynomial
\begin{equation}
    %\Delta_{2n+1}(\lambda) = (\lambda-a_1)P_n(g(\lambda))\label{eq:charPol_odd}
    \Delta_{2n+1}(\lambda) := \left(\lambda+\frac{R}{L}\right)P_n(g(\lambda))=0\label{eq:charPol_odd}
\end{equation}
if the dimensions of the $A$ matrix are odd, and 
\begin{equation}
    %\Delta_{2n}(\lambda) = P_n(g(\lambda))+b_2\cdot P_{n-1}(g(\lambda)) \label{eq:charPol_even}
    \Delta_{2n}(\lambda) := P_n(g(\lambda))-\frac{1}{LC} P_{n-1}(g(\lambda))=0 \label{eq:charPol_even}
\end{equation}
if the dimensions of $A$ are even, see~\cite{gover1994eigenproblem,da2007characteristic} for a detailed explanation\footnote{Notice that, unlike in~\cite{da2007characteristic}, in our case the upper diagonal elements are negative. Furthermore, in our $A$ matrix the upper diagonal is identical to the lower diagonal with opposite signs, although this does not seem to lead to a major simplification in the eigenvalue computation.}.

\subsection{Open Loop}
\label{sec:openLoop}
In the case of our transmission line in~\eqref{fig:PiModel}, the dimensions of the state space representation are $2n+1$, where $n$ is the number of $\pi$-sections. The equation in~\eqref{eq:charPol_odd} lead us to two nonlinear equations, using the definition in~\eqref{eq:trigDefChebyshev} and identifying the argument inside the Chebyshev polynomial $U_n$ in~\eqref{eq:defP} to $\cos(\theta)$:
\begin{eqnarray}
    \left(\lambda+\frac{R}{L}\right)\frac{1}{(LC)^n}\frac{\sin((n+1)\theta)}{\sin(\theta)} = 0\label{eq:1stEqEvals}\\
    \cos(\theta) = \frac{1}{2}LC\left(\lambda+\frac{R}{L}\right)\left(\lambda+\frac{G}{C}\right)+1 \label{eq:2ndEqEvals}    
\end{eqnarray}
%\begin{eqnarray}
%   (\lambda-a_1)(\sin((n+1)\theta) &=& 0\label{eq:xxx}\\
%   \cos(\theta) &=& \frac{g(\lambda)-b1-b2}{2\sqrt(b_1 b_2)}\\
%   (\lambda+\frac{R}{L})(\sin((n+1)\theta) &=& 0\label{eq:xxx2}\\
%    \cos(\theta) &=& \frac{g(\lambda)+\frac{2}{LC}}{frac{2}{LC}}
%\end{eqnarray}
Equation~\eqref{eq:2ndEqEvals} is quadratic in $\lambda$, so it can be solved as a function of $\cos(\theta)$. At the same time,~\eqref{eq:1stEqEvals} defines the conditions for $\theta$ plus another solution for $\lambda$, leading to the following eigenvalues:
\begin{gather} 
\lambda_0 = \frac{-R}{L},\label{lambdaReal}\\
%&\lambda_{2...n} = \frac{-R}{2L}\pm\sqrt{\frac{R^2}{4L^2}+ \frac{2}{LC}\left(1-
% \cos\left(\frac{k \pi}{n+1}\right)\right)}\\
\lambda_{\pm k} = -\frac{1}{2}\left(\frac{R}{L}+\frac{G}{C}\right)\pm\sqrt{\frac{1}{4}\left(\frac{R}{L}-\frac{G}{C}\right)^2 - \frac{2}{LC}\left(1- \cos(\theta)\right)}\notag\\ \label{lambdaComplex}
\end{gather} 
where $\theta=\frac{k \pi}{n+1}, k=1,2\hdots n$. Notice that~\eqref{eq:1stEqEvals} is well defined for $\theta=0$, and in particular $\lambda=\lambda_0$ corresponds to $\theta=0$ in~\eqref{eq:2ndEqEvals}.

A few insights can be drawn from this formula. First, in the case of a single $\pi$-section and negligible values of $R$ and $G$, the convoluted equation in~\eqref{lambdaComplex} simplifies to the well known expression of $\lambda_{\pm 1}= 1/\sqrt{LC}$ for the complex eigenvalues, as expected. Moreover, it can be easily seen how, for $n>1$, this term $1/\sqrt{LC}$ can be fairly inaccurate in describing the natural frequency of the eigenvalues~\cite{leterme2021use}. For the typical values of $R$, $L$, $G$, $C$ and $n$ the eigenvalues $\lambda_{\pm k}$ are always complex, with the imaginary part being much larger than the real part, hence the poor damping of the unloaded system. The imaginary part is always bounded by $2/\sqrt{LC}$, given that the term $2/LC$ inside the square root is much larger than the other terms. However, this bound is not very meaningful since we are mainly interested in the first resonances, as the higher modes are rarely excited. Moreover, if we increase the number of $\pi$-sections to model a line (but correspondingly reducing the values of $R$,$L$, $G$ and $C$ per section), the real part stays constant but the imaginary part of the eigenvalues increases proportional to the number of sections, thereby modelling a higher frequency range in a more accurate way. This is expected and consistent with the understanding of $\pi$-section modelling. 

%First, regardless of the number of $\pi$ sections, the imaginary part of the eigenvalues is bounded by
%, while the real part stays constant. If, on the other hand, we rescale the values in order to keep the overall L, C, G and R. 

\subsection{The effect of voltage control}
\label{sec:voltageControl}
%For the sake of readability, in the following we assume $G=0$, since typically this term is negligible. In any case it is straightforward to include this term in the eigenvalue analysis. -> I actually think i dont need this, expressions are very similar, not more cumbersome.
We consider now the impact of the voltage regulator on the eigenvalues and the damping of the grid dynamics. As mentioned before, relatively high gains may be needed in order to properly track the voltage reference. For simplicity, we consider only the voltage source $v_b$, although by symmetry the effect of the other unit will be identical. Closing the loop by means of the voltage regulator in~\eqref{eq:propCtrlVoltage} entails modifying the $(2n+1,2n)$ element of the state space matrix $A$, from $1/L$ to $1/L+k_V/L$. The state space matrix is no longer Toeplitz but nonetheless we can derive the characteristic polynomial by computing the determinant of $\lambda I -A$ along the last row using Laplace expansion:
%\begin{eqnarray*}
\begin{align}
%    \Delta_{2n+1}(\lambda) &=& (\lambda-a_1)\Delta_{2n}-(b_2+\gamma)c_2\Delta_{2n-1}\\
%    &=& (\lambda-a_1)(P_n(g(\lambda))+b_2 c_2 P_{n-1}(g(\lambda))-(b_2+\gamma)c_2\Delta_{2n-1}\\
%    &=& (\lambda-a_1)(P_n(g(\lambda))+\gamma P_{n-1}(g(\lambda))\\
    \Delta_{2n+1}(\lambda) =& \left(\lambda+\frac{R}{L}\right)\Delta_{2n}(\lambda)+\left(\frac{1}{L}+\frac{k_V}{L}\right)\frac{1}{C}\Delta_{2n-1}(\lambda)\label{eq:charPol_kV1}\notag\\
=& \left(\lambda+\frac{R}{L}\right)(P_n(g(\lambda))-\frac{1}{LC}P_{n-1}(g(\lambda))) \notag\\
&+\left(\frac{1}{L}+\frac{k_V}{L}\right)\frac{1}{C} \left(\lambda+\frac{R}{L}\right) P_{n-1}(g(\lambda))   \notag \\
=& \left(\lambda+\frac{R}{L}\right) \left(P_n(g(\lambda))+\left(\frac{k_V}{LC}\right) P_{n-1}(g(\lambda))\right)    
\end{align}
%\end{eqnarray*}
where we have used the expressions in~\eqref{eq:charPol_even} and~\eqref{eq:charPol_odd} for odd and even order. As in the open loop case in Section~\ref{sec:openLoop}, these leads to the following two nonlinear equations:
%\begin{eqnarray}
%    (\lambda-a_1)(\sin((n+1)\theta) + k_V \sin(n \theta))&=&0\label{eq:xxx}\\
%    (\lambda-a_1)\cdot \lambda - 2b_1 (1-\cos(\theta)) &=& 0 \\
%\end{eqnarray}
\begin{gather}
    \left(\lambda+\frac{R}{L}\right)\frac{\sin((n+1)\theta)+k_V \sin(n\theta)}{\sin(\theta)}  = 0\label{eq:1stEqEvalsVolt}\\
    \cos(\theta) = \frac{1}{2}LC\left(\lambda+\frac{R}{L}\right)\left(\lambda+\frac{G}{C}\right)+1 \label{eq:2ndEqEvalsVolt}   
\end{gather}
where~\eqref{eq:2ndEqEvalsVolt} is identical to~\eqref{eq:2ndEqEvals}. As before, the eigenvalues can be computed as a function of the angle $\theta$ using~\eqref{eq:2ndEqEvalsVolt}, reaching to the same expressions as in~\eqref{lambdaReal},~\eqref{lambdaComplex}, but where now the angle satisfies the following trigonometric equation:
\begin{equation}
\label{eq:kV_trig}
    \sin((n+1)\theta)+k_V \sin(n\theta)=0    
\end{equation}
Notice that again $\theta=0$ does not necessarily satisfy~\eqref{eq:1stEqEvalsVolt}, and thus $k_V$ affects all complex eigenvalues. However, the unique real eigenvalue is not affected by the control gain and satisfies~\eqref{eq:1stEqEvalsVolt} for $\lambda=-R/L$ and~\eqref{eq:2ndEqEvalsVolt} with $\theta=0$. We can clearly see that $k_V$ only affects the angle $\theta$, and therefore, under typical values of $L$, $C$, $G$, $R$ and $n$, it does not influence the real component but only the imaginary part of the eigenvalues in~\eqref{lambdaComplex}. Hence, it can be already hinted that the proportional control influences negatively the damping of the system. Even though the relationship between the controller gain and $\theta$ is not trivial, we can actually prove that the term $1-\cos(\theta)$ is monotonically non-decreasing with respect to $k_V$, as stated in the following proposition:
%\rr{Can I actually show it is always increasing, except at 0?}

\begin{proposition}
\label{pro1}
%Consider a dynamic system as in~ . Under typical transmission lines parameters (i.e., $R<<L$), the proportional controller only influences the imaginary part of the complex eigenvalues:
%\begin{eqnarray}
%    \frac{\partial}{\partial k_V} Im\{\lambda\}>0,\qquad  \frac{\partial}{\partial k_V} Re\{\lambda\}=0
%\end{eqnarray}
Consider $\theta(k_V)$ as implicitly defined by~\eqref{eq:kV_trig}. Then,
\begin{equation}
     \frac{d}{d k_V} (1-\cos(\theta))\geq 0,\quad \forall \:\theta \in[0,\pi[
\end{equation}
\end{proposition}
\begin{proof}
We take as starting point the implicit relationship between $\theta$ and the gain $k_V$ in~\eqref{eq:kV_trig}, and differentiate with respect to $k_V$
\begin{align*}
&\frac{d}{dk_V} \left( \sin((n+1)\theta)+k_V \sin(n\theta)\right) =  \\
&\cos((n+1)\theta)(n+1)\frac{d\theta}{dk_V}+\sin(n\theta)+k_v \cos(n\theta) n \frac{d\theta}{dk_V} =0
\end{align*}
By replacing $k_V$ using again~\eqref{eq:kV_trig} we obtain a closed expression for $d\theta / dk_V$:
\begin{align}
&\frac{d\theta}{dk_V} =\label{eq:dtheta}\\
&\frac{-\sin(n\theta)^2}{(n+1)\cos((n+1)\theta)\sin(n\theta) - n\cos(n\theta)\sin((n+1)\theta)}\notag
\end{align}
The numerator is clearly non-positive, and the denominator can be shown to be always non-positive using trigonometric identities and the fact that $\sin(n\theta)\leq n\sin(\theta)$ (omitted here for conciseness). Moreover, using L'H\^{o}pital's rule it can be shown that~\eqref{eq:dtheta} is well defined for $\theta\in[0,\pi[$, since the the denominator is only 0 at $\theta=0$. Therefore, $d\theta / dk_V\geq 0$. Finally, we compute the dependence of $1-\cos(\theta)$ as a function of the controller gain:
\begin{equation*}
   \frac{d}{dk_V}  (1-\cos(\theta)) = \sin(\theta) \frac{d\theta}{dk_V}    
\end{equation*}
which is also greater or equal than 0 in the domain of $\theta$, see~\eqref{eq:trigDefChebyshev}.
\end{proof}

\begin{remark}
    The zeros of $\frac{d}{d k_V} (1-\cos(\theta))$ correspond to $\theta=0$ and $\sin(n\theta)=0$. In both cases, it can be seen that these values of $\theta$ do not satisfy~\eqref{eq:1stEqEvalsVolt} and~\eqref{eq:2ndEqEvalsVolt} for the complex eigenvalues, and therefore for all values of $\theta$ of interest the derivative is strictly positive.
\end{remark}
Given that the square root is a monotonic function, it can be stated that increasing $k_V$ exclusively increases the imaginary part of the complex eigenvalues, that is:
\begin{eqnarray}
    \frac{\partial}{\partial k_V} \mathfrak{Im}\{\lambda\}>0,\qquad  \frac{\partial}{\partial k_V} \mathfrak{Re}\{\lambda\}=0
\end{eqnarray}
This result clearly shows the existing trade-off between damping and tracking error. 
\begin{remark}
\label{remarkConjecture}
    An explicit solution for~\eqref{eq:kV_trig} can be easily found for $k_V=1$. In that case, using simple trigonometric identities it is straightforward to see that:
    \begin{equation}
    \sin((n+1)\theta) + \sin(n\theta) = 2 \sin \left( \left(n+\frac{1}{2} \right) \theta \right) \cos\left(\frac{\theta}{2}\right) = 0
    \end{equation}
    Hence, $\theta=\frac{k\pi}{n+1/2}$, for $k=1,2,\hdots n$ is a solution, which is indeed a larger angle than in the open loop, as expected from the previous discussion. For the general case, it has been numerically observed that $\theta=\frac{k\pi}{n+1/(1+k_V)}$ represents a fairly accurate approximation. Interestingly, this implies that, for large values of $k_V$, the eigenvalues converge to the open loop eigenvalues of a model with $n-1$ $\pi$-sections. An explicit expression for the angle of the form $\theta=\frac{k\pi}{n+1/g(k_V)}$ (for some monotonically increasing function $g$) may not exist but will be the focus of future research. 
\end{remark}

%Can I compute the error in steady state with a proportional controller? It should be easy -> maybe not so much, there is an inverse of A that is very tricky.

%SO WITH MANY PI SECTIONS THEN THE INFLUENCE OF THE ANGLE IS MUCH SMALLER?? SO THE INFLUENCE OF KV IS MUCH SMALLER?? We see that in the dtheta/dk plot (or do we see the opposite, since what i plotted was the denominator? so the evals that move the most are the ones with small and large theta?) -> tricky discussion

\subsection{The effect of virtual resistance and other variables}
The approach in the previous subsection can be applied to other grid parameters and other control loops. For instance, given that the proportional controller decreases the (already low) damping of the system, a black start unit could somehow incorporate another loop that compensates the effect of the voltage feedback loop, using available local measurements. In particular, the voltage control can be expanded as:
\begin{equation}
\label{eq:virtual_resistance}
    v_b = k_V (v_{ref}-v_{s-1}) + k_I i_{s}     
\end{equation}
where the converter also measures the current at its terminal\footnote{The virtual resistance appears as a positive feedback loop in our mathematical formulation because of the sign convention used for currents in our modelling.} and modifies the voltage accordingly. Intuitively speaking, this can be seen as a virtual resistance loop, which has been already suggested for black start purposes in~\cite{jain2022virtual}, although the focus was on reducing inrush and sympathetic currents. From an algebraic perspective, this increases the magnitude of one of the diagonal elements and therefore it is expected to improve the overall damping. In this case, the virtual resistance modifies the last diagonal term in the $A$ matrix (that is, the $(2n+1,2n+1)$ term), from $-R/L$ to $-R/L-k_I/L$. 

For simplicity, we analyze here just the effect of $k_I$. The combined effect of $k_V$ and $k_I$ can also be computed using similar strategies. Following a similar approach as in Section~\ref{sec:voltageControl}, the characteristic polynomial under the presence of a virtual resistance can also be derived as before:
\begin{align}
    \Delta_{2n+1}(\lambda) =& \left(\lambda+\frac{R}{L}+\frac{k_I}{L}\right)\Delta_{2n}(\lambda)+\left(\frac{1}{LC}\right)\Delta_{2n-1}(\lambda)
    \notag\\%\label{eq:charPol_kI1}
=& \left(\lambda+\frac{R}{L}+\frac{k_I}{L}\right)\left(P_n(g(\lambda))-\frac{1}{LC}P_{n-1}(g(\lambda))\right) \notag\\
&+\left(\frac{1}{LC}\right) \left(\lambda+\frac{R}{L}\right) P_{n-1}(g(\lambda))   \notag \\
=& \left(\lambda+\frac{R+k_I}{L}\right) P_n(g(\lambda))-\frac{k_I}{L^2C} P_{n-1}(g(\lambda))  \label{eq:charPolCurr}
\end{align}
Similar expressions can be found in\cite{da2007characteristic}. From this characteristic polynomial, we can again derive one nonlinear equations of the form:
\begin{equation}
    \left(\lambda+\frac{R}{L}+\frac{k_I}{L}\right)\sin((n+1)\theta)-\frac{k_I}{L}\sin(n\theta)=0 \label{eq:1stEqEvalsCurr}
\end{equation}
%\rr{Is there a plus or a minus before kI??}
plus the equation for $\cos(\theta)$ as in~\eqref{eq:2ndEqEvalsVolt}. In this case the equations are not partially decoupled as before, since the control gain $k_I$ does not affect exclusively the angle $\theta$, so the derivation becomes more cumbersome. We focus on the complex eigenvalues, which are the ones of interest in order to determine the damping in the system. It can be seen in~\eqref{eq:1stEqEvalsCurr} that the corresponding $\theta$ needs to be also complex for~\eqref{eq:1stEqEvalsCurr} to hold. The fact that the angle is complex already hints that changes in the angle will now affect the real part of the eigenvalues. 

To proceed, we can carry out a series expansion for~\eqref{eq:1stEqEvalsCurr}, given that the function is also analytic in the complex plane. It can be seen as a function of $3$ variables, namely $\theta$, $k_I$ and $\lambda$ evaluated at the operating point corresponding to $k_{I_*}=0$, with the corresponding values for $\theta_*$ and $\lambda_*$ defined in Section~\ref{sec:openLoop}:

\begin{small}
\begin{align}
    &F_1 = \left(\lambda+\frac{R}{L}+\frac{k_I}{L}\right)\sin((n+1)\theta)-\frac{k_I}{L}\sin(n\theta)\notag\\ %\label{eq:1stEqEvalsCurr}
    & \approx F_1(\lambda_*,k_{I_*},\theta_*) + \frac{\partial F_1}{\partial \theta}\Bigr|_{\substack{\theta_*\\\lambda_*\\k_{I_*}}} \Delta \theta 
    + \frac{\partial F_1}{\partial k_I}\Bigr|_{\substack{\theta_*\\\lambda_*\\k_{I_*}}} \Delta k_I + 
    \frac{\partial F_1}{\partial \lambda}\Bigr|_{\substack{\theta_*\\\lambda_*\\k_{I_*}}} \Delta \lambda =\notag\\ 
    &  \left(\left(\lambda_*+\frac{R}{L}+\frac{k_{I_*}}{L}\right)(n+1)\cos((n+1)\theta_*) -\frac{k_{I_*}}{L}n\cos(n\theta_*)\right)\Delta \theta \notag\\
    &+\frac{1}{L}\left(\sin((n+1)\theta_*)-\sin(n\theta_*)\right)\Delta k_I 
    + \sin((n+1)\theta_*)\Delta \lambda = 0
\end{align}
\end{small}
Taking into account that the operating point satisfies $F_1=0$ and therefore $\sin((n+1)\theta_*=0$, we can reach a simple relationship between variations in the control gain $k_I$ and the angle $\theta$:
\begin{equation}
    \Delta \theta = -\frac{1}{L} \frac{\sin\left(\frac{n}{n+1}\pi k\right)}{(n+1)\left(\lambda_* + \frac{R}{L}\right)} \Delta k_I
\end{equation}
for $k=1,2,\hdots n$. Given that $\lambda_*$ is mainly imaginary, and the rest of the terms are real, it can be seen that a variation in the control gain leads mainly to a change in the imaginary part of $\theta$. Moreover, the sensitivity is always positive, that is, the imaginary magnitude of the angle always increases with $\Delta k_I$. Likewise, we can now look at the relationship between the angle and the eigenvalues from~\eqref{eq:2ndEqEvals}:
\begin{align}
    &F_2 = \cos(\theta) - \frac{1}{2}LC\left(\lambda+\frac{R}{L}\right)\left(\lambda+\frac{G}{C}\right) -1 \notag\\
    & \approx F_2(\lambda_*,\theta_*) + \frac{\partial F_2}{\partial \theta}\Bigr|_{\substack{\theta_*\\\lambda_*}} \Delta \theta +
    \frac{\partial F_2}{\partial \lambda}\Bigr|_{\substack{\theta_*\\\lambda_*}} \Delta \lambda \notag\\ 
    & = -\sin(\theta_*)\Delta \theta - LC \left(\lambda_* +\frac{1}{2}\left(\frac{G}{C}+\frac{R}{L}\right)\right)\Delta \lambda  = 0
\end{align}
Plugging the value for $\lambda_*$ it can be easily seen that an increase of the imaginary value in $\Delta \theta$ leads to an increase in the real magnitude of $\Delta \lambda$, and therefore the control gain positively affects the damping of the system. Unlike in Section~\ref{sec:voltageControl}, this is a local result that does not hold for extremely large values of $k_I$. This local result is nonetheless of interest, given the expected range of control gain values. 

The derived results for $k_V$ and  $k_I$ hint at two clear statements for the behaviour of the grid-forming unit energizing the grid: large values of $k_V$ improve the tracking accuracy of the voltage reference at the cost of reducing the damping of the resonances, but this effect can be compensated by increasing the values of $k_I$.

%While a detailed demonstration involves a lengthy complex analysis, we informally show here how $k_I$ affects the real part of the eigenvalues. 

%First, we write the angle as $\theta = \theta_R + j\theta_I$. Then, by computing local sensitivities of $d\theta_R /d k_V$ and $d\theta_I /d k_V$ (as in the previous section), it can be shown that an increase in the control gain barely influences $\theta_R$, but it leads to a significant increase in $\theta_I$. We now focus again on the expression in~\eqref{lambdaComplex}, where the eigenvalues depend on $\sqrt{1-\cos \theta}$. 

%Given that the angle is now imaginary, this term corresponds to the square root of a complex number. It can be further shown that an increase in $\theta_I$ leads to a increase in the real part of principal solution of $\sqrt{1-\cos \theta}$, thereby increasing the damping in the system. A complete and formal demonstration will be included in future work.

Before finalizing this section, we would like to emphasize the potential of this framework to analyze the effect of other design choices or parameters. For instance, we can consider a control strategy of the voltage source $v_b$ that can also read the current at the other end of the line $i_1$. This can be seen as a communication between two voltage sources attempting at energizing the same transmission line, in order to improve coordination. In this case, the element $(2n+1,1)$ is no longer zero but instead $k_X/L$, for some control gain $k_X\geq 0$. Using the same formulation as before, the characteristic polynomial of the system under this controller turns out to be:
\begin{align*}
    &\Delta_{2n+1}(\lambda) := \left(\lambda+\frac{R}{L}\right) P_n(g(\lambda)) +  \frac{k_X}{L}\frac{1}{(LC)^n}  \notag\\
    &=\left(\lambda+\frac{R}{L}\right) \frac{1}{(LC)^n} \frac{\sin((n+1)\theta)}{\sin(\theta)} + \frac{k_X}{L} \frac{1}{(LC)^n} = 0
\end{align*}
This expression can be further analyze the derive the influence of $k_X$ on the real part of the eigenvalues. For instance, it is straightforward to see that $-R/L$ is no longer an eigenvalue of the system.

Likewise, using this framework the impact of a virtual reactance can be analysed. The effect of equipping the controller with a virtual reactance is somewhat similar to the combined effect of a voltage control and a current control, as expected from the closed loop matrix structure, hence increasing the real and the imaginary components of the eigenvalues. Other actions that imply minor changes in the matrix structure have also been analysed using this framework, for instance topological changes. Preliminary results show that adding loops to a radial topology does not decrease the damping of the first resonance in the grid.

\section{An academic example}
To illustrate our results, we consider a \SI{100}{km} transmission line with R = $0.02$ \si{\Omega / km}, L = $0.5 \cdot 10^{-3}$ \si{H/km}, C = $0.4 \cdot 10^{-6}$ \si{F/km} and G = $0$ \si{S/km}. These values refers to a typical \SI{110}{kV} transmission line. The line is divided in 6 $\pi$-section of equal length. Therefore, the system has $13$ states variables.

We first evaluate the root locus of the system with two voltage sources at each extreme of the line. For ease of explanation, only $v_b$ is equipped first with a voltage controller and then with a current loop. Fig. \ref{fig:rootLocus_kv} and \ref{fig:rootLocus_ki} show the evolution of the eigenvalues as a function of the proportional gain $k_V$ and the virtual resistance $k_I$ respectively. As expected from the theoretical results above, independently of the values of $k_V$ and $k_I$, the system has a pure real eigenvalue and 5 complex eigenvalues pairs. The increase of $k_V$ increases the imaginary part, thereby reducing the damping. On the other hand, the increase of $k_I$ shifts all the eigenvalue to the left increasing the damping. It is important to emphasize the different sensitivities of each eigenvalue: in the case of $k_V$, the smallest eigenvalues (the most relevant ones for our application) do not vary vastly as a function of the control gain. This could also be deducted from the expressions in Remark~\ref{remarkConjecture} and its conjectured explicit dependence on $k_V$: for $n=6$ and $k=1$ (the first complex eigenvalue), the effect of $k_V$ on the angle $\theta$ is minimal. In the case of $k_I$, the intermediate eigenvalues exhibit larger sensitivity than the other ones. This is however dependent on the line parameters.
\begin{figure}[h]
%\hrule
\includegraphics[width=\columnwidth]{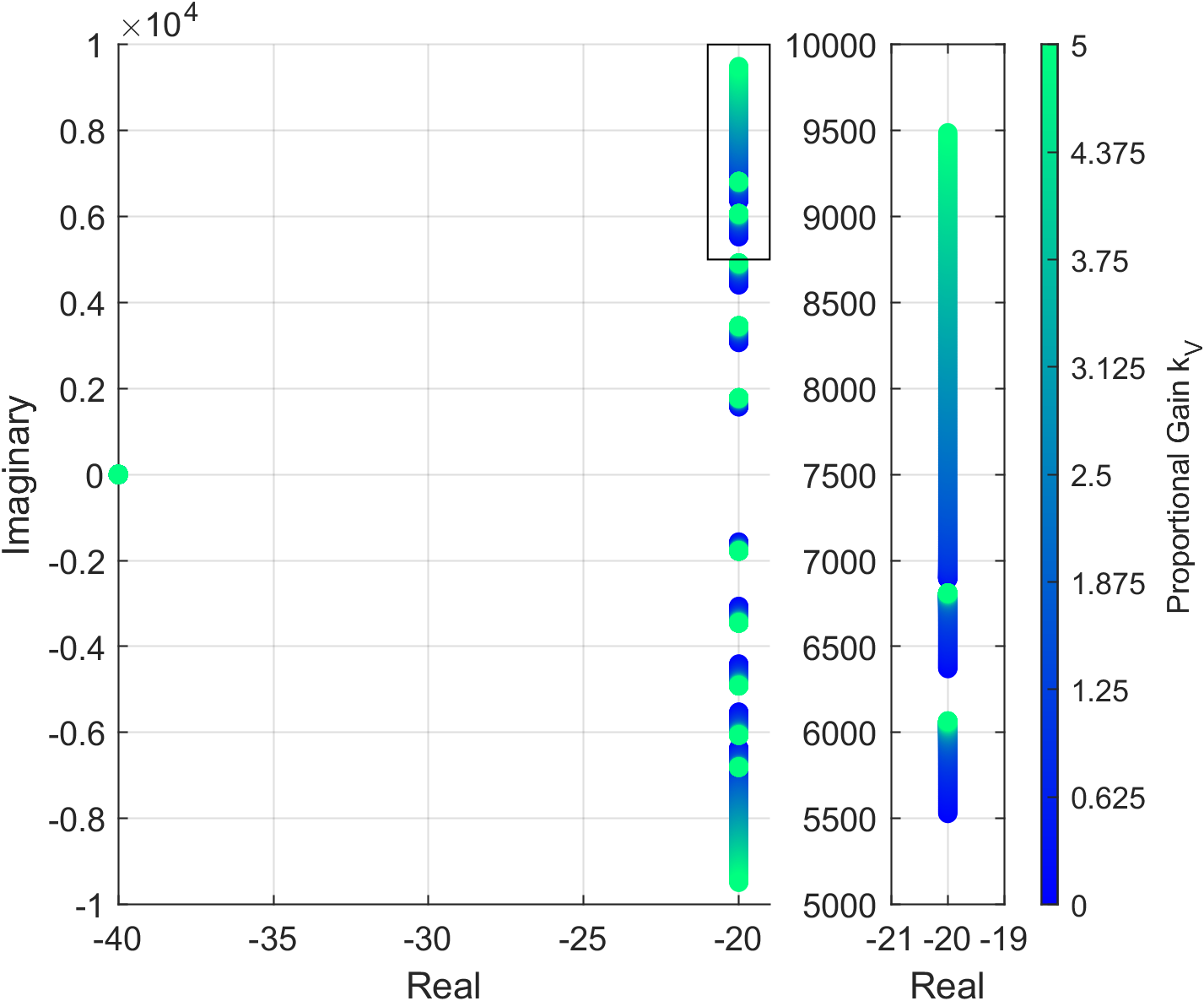}
\caption{ Root locus of~\eqref{eq:modelPi} as a function of $k_V$}
\label{fig:rootLocus_kv}
\end{figure}

\begin{figure}[h]
%\hrule
\includegraphics[width=\columnwidth]{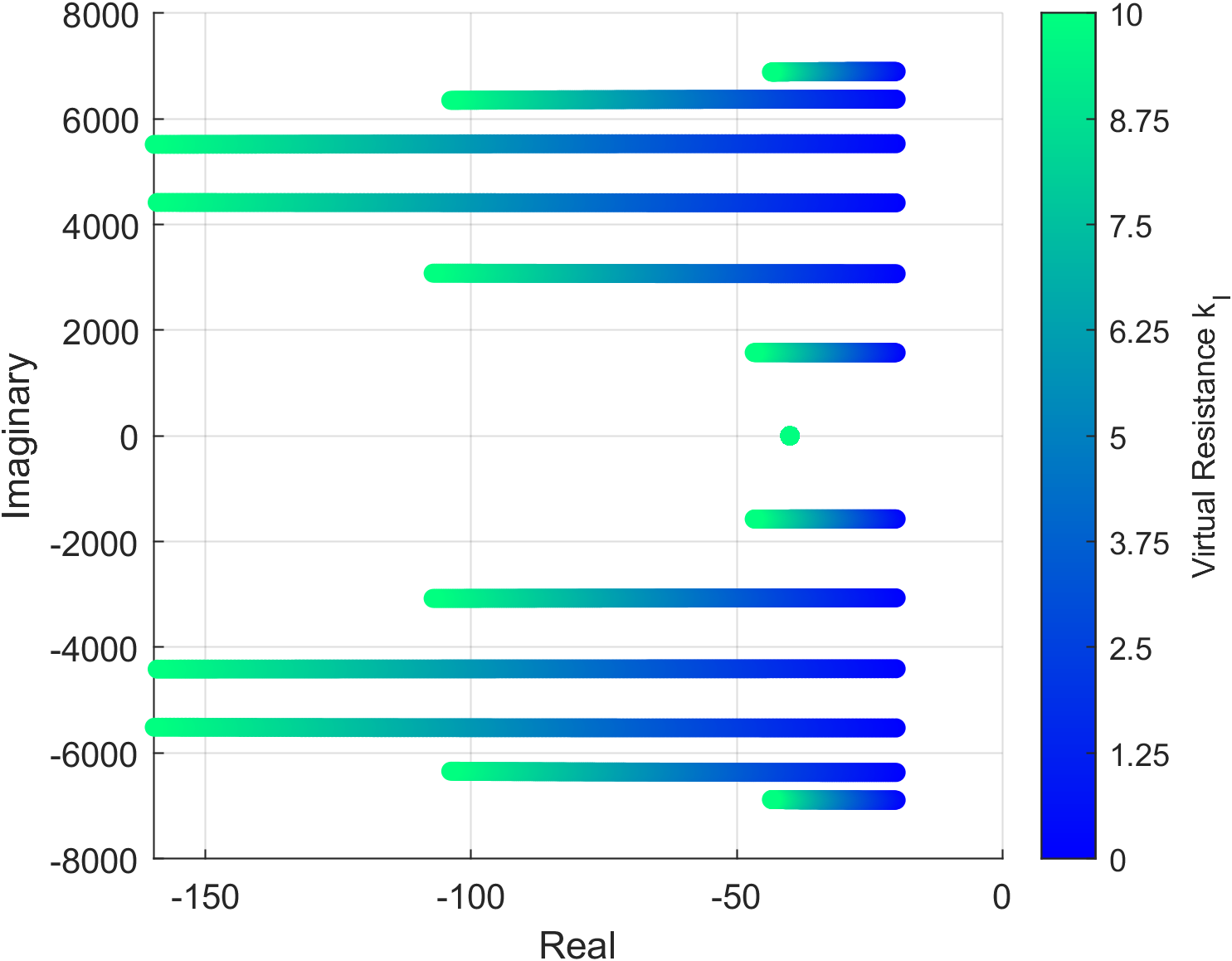}
\caption{ Root locus of~\eqref{eq:modelPi} as a function of $k_I$}
\label{fig:rootLocus_ki}
\end{figure}

Similar effects can be seen looking at the frequency response of the system. Figure~\ref{fig:bode} displays the frequency response between the voltage source and the voltage in the intermediate node $v_3$. While $k_V$ does not reduce the damping, it does shift the resonances towards higher frequencies. On the other hand, the effect of $k_I$ is to reduce the resonances in the system.

We also show in Fig.~\ref{fig:resonances} the time evolution of the voltage at the third node $v_3$ under different control gains. The reference for both voltage sources is a \SI{50}{Hz} sine wave\footnote{The voltage sources $v_a$ and $v_b$ are assumed to be perfectly synchronized, given that our modelling does not include a synchronization algorith.}. The controlled voltage source $v_b$ injects a \SI{250}{Hz} component, of a magnitude equal to $2\%$ of the main \SI{50}{Hz} component. In the open loop case, the \SI{250}{Hz} component is clearly amplified, leading to a worse harmonic distortion. The presence of the voltage controller leads to a lower distortion, but not because of improved damping, but rather because the resonances are moved up. Nonetheless, this strategy of shifting the frequencies cannot reliably avoid resonances, given the general uncertainties in the location of resonances and harmonics injection. On the other hand, $k_I$ reduces the distortion by overall increasing the damping, leading to a more robust strategy. %Finally, using this same model it has also been observed how the virtual resistance leads to a faster decay of transients upon a hard energization, compared to the open loop case and the case of the voltage controller.

Finally, we analyze the impact of the virtual resistance in the energization of a grid. We consider the well-known 9-bus system, where one inverter energizes the unloaded grid. Each line is $100$km long and is modelled using $3$ $\pi$-sections, with the same line parameters as before. Fig.~\ref{fig:energization_9bus} shows the current at one node when the energization is carried out at $t=0.1s$ in open loop (that is, the inverter generates a fixed voltage profile without feedback) and under the presence of a voltage controller and a virtual resistance as in~\eqref{eq:virtual_resistance}, with $k_V=10$ and $k_I=100$. The inverter not only injects the fundamental 50Hz component but also odd harmonics up to the $9^\text{th}$ component, with an amplitude of 1\% of the main component. Not only we see a large overshoot in current at the beginning, but also in steady state the undamped resonances lead to larger steady state current values. This implies either an unnecessary overdimensioning of the inverter, or the need to reduce the initial grid topology to be energized by the inverter. Similar behaviour has been observed in other nodes across the grid and in the voltage waveforms, leading to dangerous overvoltages that jeopardize the black start of the system. The time response of the system with only the voltage controller active is very similar to the open loop response and thus has been omitted. 

% While (on absolute terms) small values of $k_V$ do not affect much the imaginary part of the smallest eigenvalues, $k_I$ does affect the most the real parts of $\lambda_{\pm 1}$. This implicitly shows how the gain does not deteriorate massively the damping, but at the same time resonances cannot move up to higher frequencies that would not be excited by the inverters. 

%Do I want to show how/if things vary if the source impedance is relevant? And not ideal matrix.
%Can I show tracking error for the ramp? (for different gains)
%Do I want to show here deriv(1-cos theta)/dkv in a plot? To see the effect on different frequencies

%Check if in our example the resonances are close to the harmonics (in one paper they showed it was close to the 4th and therefore overvoltages). Some say above 800Hz dont bother

\begin{figure}[h]
%\hrule
\includegraphics[width=\columnwidth]{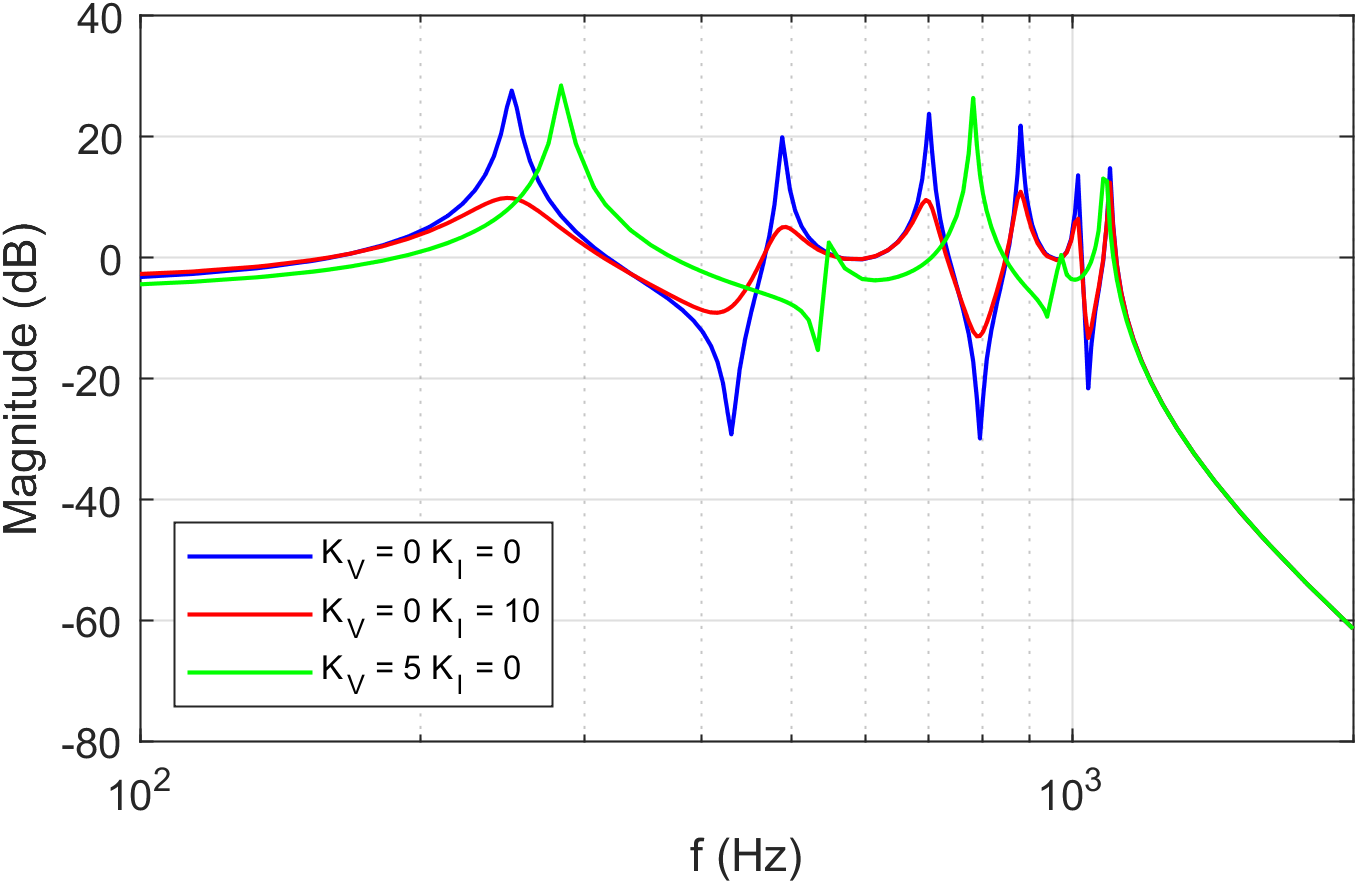}
\caption{Effect of the controller gains on the frequency response}
\label{fig:bode}
\end{figure}

%\rr{It has also been observed that, under a hard energization, lines with low resistance values, 
%Diego:s results: for very low values of R, even a simple step in voltage excites high frequencies that take forever to damp }

%Point out that stability of that type of dynamical system (in dq coords, with G, R in the diag as in Florian‘s paper) has been widely studied. The evals not so much…

%I need to take a look at the definition of square roots of complex numbers, it is not trivial

%Better than colours for each term in the matrix, just show root locus for each term? with a massive subplot. I should do it with one line and say that with other parameters the root locus looks very similar

\begin{figure}[h!]
%\hrule
\includegraphics[width=\columnwidth]{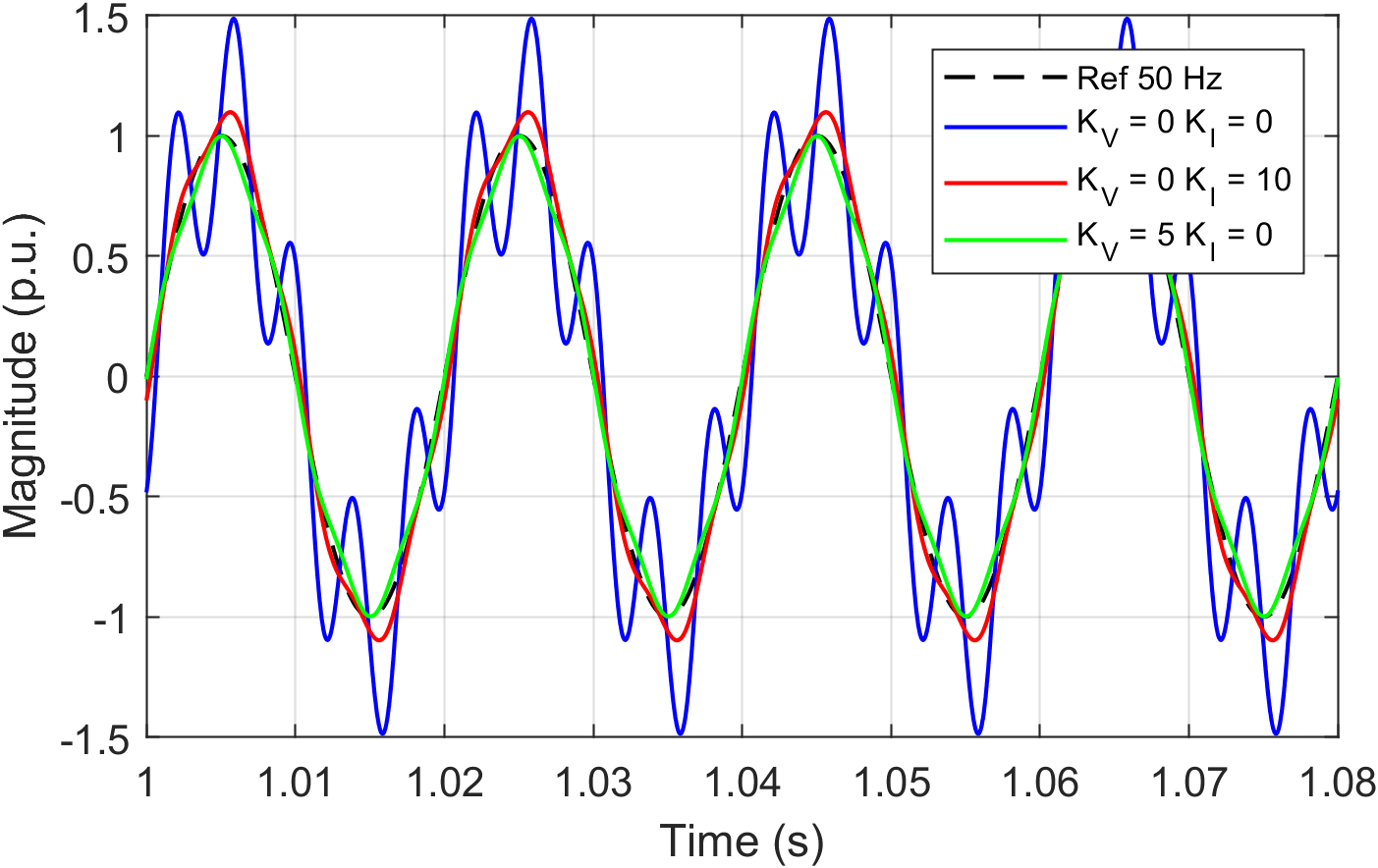}
\caption{ Time evolution of the voltage under harmonic injection}
\label{fig:resonances}
\end{figure}

\begin{figure}[h!]
%\hrule
\centering
\includegraphics[width=1\columnwidth]{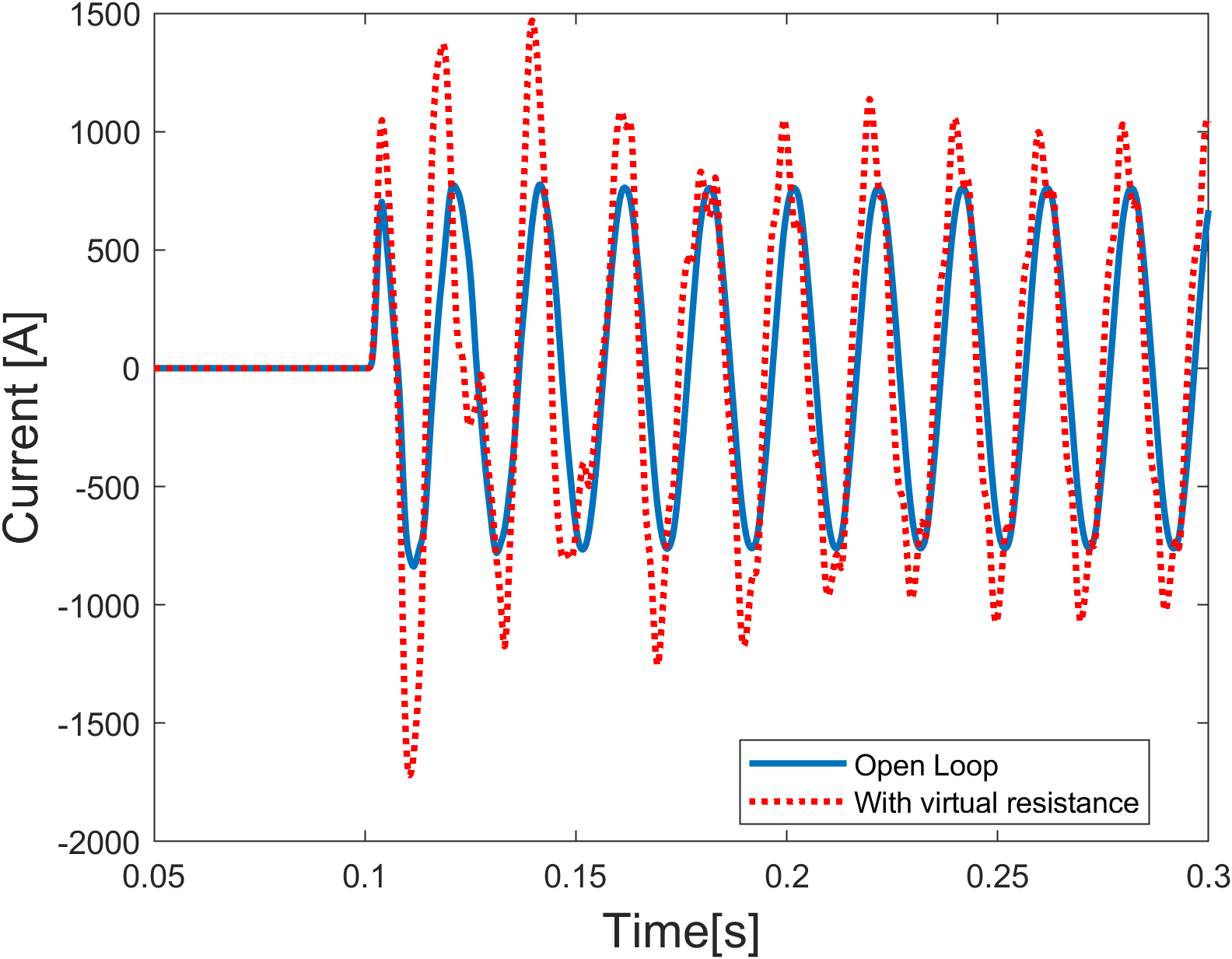}
\caption{Energization of the 9-bus system with and without virtual resistance}
\label{fig:energization_9bus}
\end{figure}

\section{Outlook}
In this paper we have explored the energization of a long transmission line during black start. In particular, the damping properties and eigenvalues of the system under different controllers have been analyzed. We hope that the employed framework, drawing parallelisms with Toeplitz matrices and Chebyshev polynomials, can be further utilized to understand the impact of other parameters and design choices, such as the output filters of the energizing filters (which are known to create resonances at low frequencies), the  network topology, the impact of a virtual reactance, the location of loads~\cite{cpes2024load} or the presence of transformers at the end of the lines, for instance. The end goal of this work is twofold: first, help defining a set of requirements for black-start capable grid-forming inverters; second, construct a set of design rules for a successful black start strategy relying on a backbone energization.

\nocite{dorfler2018electrical,yueh2008explicit,bastawrous2022closed,da2020comments,alvarez2005some,nikkila2020fast}

\bibliographystyle{IEEEtran}
\bibliography{resonancesBlackstart_bib}

\begin{thebibliography}{10}
\providecommand{\url}[1]{#1}
\csname url@samestyle\endcsname
\providecommand{\newblock}{\relax}
\providecommand{\bibinfo}[2]{#2}
\providecommand{\BIBentrySTDinterwordspacing}{\spaceskip=0pt\relax}
\providecommand{\BIBentryALTinterwordstretchfactor}{4}
\providecommand{\BIBentryALTinterwordspacing}{\spaceskip=\fontdimen2\font plus
\BIBentryALTinterwordstretchfactor\fontdimen3\font minus \fontdimen4\font\relax}
\providecommand{\BIBforeignlanguage}[2]{{%
\expandafter\ifx\csname l@#1\endcsname\relax
\typeout{** WARNING: IEEEtran.bst: No hyphenation pattern has been}%
\typeout{** loaded for the language `#1'. Using the pattern for}%
\typeout{** the default language instead.}%
\else
\language=\csname l@#1\endcsname
\fi
#2}}
\providecommand{\BIBdecl}{\relax}
\BIBdecl

\bibitem{noris2019power}
L.~Noris, J.~Rueda, E.~Rakhshani, and A.~Korai, ``Power system black-start and restoration with high share of power-electronic converters,'' in \emph{2019 IEEE Power \& Energy Society General Meeting (PESGM)}.\hskip 1em plus 0.5em minus 0.4em\relax IEEE, 2019, pp. 1--5.

\bibitem{ng_restart}
{National Grid}, ``Demonstration of black start from ders,'' Technical report[online]. Accessed: 2024/02/14. Available at: https://www.nationalgrideso.com/future-energy/projects/distributed-restart, 2021.

\bibitem{jain2020blackstart}
H.~Jain, G.-S. Seo, E.~Lockhart, V.~Gevorgian, and B.~Kroposki, ``Blackstart of power grids with inverter-based resources,'' in \emph{IEEE Power \& Energy Society General Meeting (PESGM)}, 2020.

\bibitem{aniceto2023towards}
J.~M. Aniceto, J.~Bosch, and A.~Anta, ``Towards a full automatizated ultrafast black start through inverter-based generation,'' in \emph{2023 IEEE Belgrade PowerTech}.\hskip 1em plus 0.5em minus 0.4em\relax IEEE, 2023, pp. 1--6.

\bibitem{jain2021blackstart}
A.~Jain, O.~Sabor{\'\i}o-Romano, J.~N. Sakamuri, and N.~A. Cutululis, ``Blackstart from hvdc-connected offshore wind: Hard versus soft energization,'' \emph{IET Renewable Power Generation}, vol.~15, no.~1, pp. 127--138, 2021.

\bibitem{saldana2001methodology}
C.~Saldana and G.~Calzolari, ``Methodology utilized in black-start studies on ehv power networks,'' in \emph{IV International Conference on Power System Transients, IPST}, 2001.

\bibitem{emin2014transformer}
Z.~Emin, M.~Martinez-Duro, M.~V. Escudero, R.~Adams, H.~S. Bronzeado, B.~Caillault, N.~Chiesa, D.~Jacobson, L.~Kocis, T.~Martinich \emph{et~al.}, ``Transformer energization in power systems: A study guide,'' \emph{CIGRE WG C}, vol.~4, 2014.

\bibitem{macias2005comparison}
J.~R. Macias, A.~G. Exposito, and A.~B. Soler, ``A comparison of techniques for state-space transient analysis of transmission lines,'' \emph{IEEE Transactions on Power Delivery}, vol.~20, no.~2, pp. 894--903, 2005.

\bibitem{dorfler2018electrical}
F.~D{\"o}rfler, J.~W. Simpson-Porco, and F.~Bullo, ``Electrical networks and algebraic graph theory: Models, properties, and applications,'' \emph{Proceedings of the IEEE}, vol. 106, no.~5, pp. 977--1005, 2018.

\bibitem{mallada2011improving}
E.~Mallada and A.~Tang, ``Improving damping of power networks: Power scheduling and impedance adaptation,'' in \emph{2011 50th IEEE Conference on Decision and Control and European Control Conference}.\hskip 1em plus 0.5em minus 0.4em\relax IEEE, 2011, pp. 7729--7734.

\bibitem{borsche2015effects}
T.~S. Borsche, T.~Liu, and D.~J. Hill, ``Effects of rotational inertia on power system damping and frequency transients,'' in \emph{2015 54th IEEE conference on decision and control (CDC)}.\hskip 1em plus 0.5em minus 0.4em\relax IEEE, 2015, pp. 5940--5946.

\bibitem{tayyebi2022system}
A.~Tayyebi, A.~Magdaleno, D.~Vettoretti, M.~Chen, E.~Prieto-Araujo, A.~Anta, and F.~D{\"o}rfler, ``System-level performance and robustness of the grid-forming hybrid angle control,'' \emph{Electric Power Systems Research}, vol. 212, p. 108503, 2022.

\bibitem{gover1994eigenproblem}
M.~J. Gover, ``The eigenproblem of a tridiagonal 2-toeplitz matrix,'' \emph{Linear Algebra and its Applications}, vol. 197, pp. 63--78, 1994.

\bibitem{al2017pole}
H.~Al~Ba'ba'a, M.~Nouh, and T.~Singh, ``Pole distribution in finite phononic crystals: Understanding bragg-effects through closed-form system dynamics,'' \emph{The Journal of the Acoustical Society of America}, vol. 142, no.~3, pp. 1399--1412, 2017.

\bibitem{bastawrous2022closed}
M.~V. Bastawrous and M.~I. Hussein, ``Closed-form existence conditions for bandgap resonances in a finite periodic chain under general boundary conditions,'' \emph{The Journal of the Acoustical Society of America}, vol. 151, no.~1, pp. 286--298, 2022.

\bibitem{da2007characteristic}
C.~da~Fonseca, ``The characteristic polynomial of some perturbed tridiagonal k-toeplitz matrices,'' \emph{Appl. Math. Sci}, vol.~1, no.~2, pp. 59--67, 2007.

\bibitem{leterme2021use}
W.~Leterme, E.~Heylen, and D.~Van~Hertem, ``On the use of the frequency domain in assessing resonant overvoltages during transformer energization.'' \emph{Electric Power Systems Research}, vol. 196, p. 107164, 2021.

\bibitem{jain2022virtual}
A.~Jain, O.~Sabor{\'\i}o-Romano, J.~N. Sakamuri, and N.~A. Cutululis, ``Virtual resistance control for sequential green-start of offshore wind power plants,'' \emph{IEEE Transactions on Sustainable Energy}, vol.~13, no.~3, pp. 1420--1429, 2022.

\bibitem{cpes2024load}
A.~Anta and D.~Cifelli, ``The impact of load placement on grid resonances during grid restoration,'' \emph{12th IFAC Symposium on Control of Power \& Energy Systems}, 2024.

\bibitem{yueh2008explicit}
W.-C. Yueh and S.~S. Cheng, ``Explicit eigenvalues and inverses of tridiagonal toeplitz matrices with four perturbed corners,'' \emph{the ANZIAM Journal}, vol.~49, no.~3, pp. 361--387, 2008.

\bibitem{da2020comments}
C.~M. da~Fonseca, ``Comments on the spectrum of a tridiagonal k-toeplitz matrix,'' \emph{Journal of Computational and Applied Mathematics}, vol. 375, p. 112793, 2020.

\bibitem{alvarez2005some}
R.~{\'A}lvarez-Nodarse, J.~Petronilho, and N.~R. Quintero, ``On some tridiagonal k-toeplitz matrices: Algebraic and analytical aspects. applications,'' \emph{Journal of computational and applied mathematics}, vol. 184, no.~2, pp. 518--537, 2005.

\bibitem{nikkila2020fast}
A.-J. Nikkil{\"a}, A.~Kuusela, M.~Weixelbraun, L.~Haarla, M.~Laasonen, and A.~Pahkin, ``Fast restoration of a critical remote load area using a gradual voltage build-up procedure,'' \emph{IET Generation, Transmission \& Distribution}, vol.~14, no.~7, pp. 1320--1328, 2020.

\end{thebibliography}

\end{document}